\documentclass[a4paper, accepted=2025-04-16, onecolumn ]{quantumarticle}
\pdfoutput=1
\usepackage[numbers,sort&compress]{natbib}
\usepackage[margin=1in]{geometry}
\usepackage[utf8]{inputenc}
\usepackage{amsmath, mathtools}
\usepackage{amsthm}
\usepackage{amsfonts}
\usepackage{braket}
\usepackage{tikz}
\usetikzlibrary{quantikz2}
\usepackage{hyperref}
\usepackage{thmtools, thm-restate}
\usepackage{glossaries}
\usepackage[automake]{glossaries-extra}
\newtheorem{theorem}{Theorem}
\newtheorem{corollary}[theorem]{Corollary}
\newtheorem{lemma}[theorem]{Lemma}
\newtheorem{definition}[theorem]{Definition}

\newcommand{\ketbra}[2]{|#1\rangle\!\langle #2|}

\title{Quantum and classical algorithms for nonlinear unitary dynamics}
\author{Noah Br\"ustle}
\affiliation{Department of Computer Science, University of Toronto, Toronto, Canada}
\author{Nathan Wiebe}
\affiliation{Department of Computer Science, University of Toronto, Toronto, Canada}
\affiliation{Pacific Northwest National Laboratory, Richland, USA}
\affiliation{Canadian Institute for Advanced Research, Toronto, Canada}
\date{}

\newglossary{symbols}{sym}{sbl}{List of Abbreviations and Symbols}
\makeglossaries

\newglossaryentry{ket}{type=symbols,name={\ensuremath{\ket{\psi}}},
description={A (vertical) vector in the numbers. This may represent a pure quantum state}}
\newglossaryentry{bra}{type=symbols,name={\ensuremath{F^{n^\ast}}},
description={$n$-fold convolution of the distribution function/distribution $F$}}

\begin{document}

\maketitle
\begin{abstract}
    Quantum algorithms for Hamiltonian simulation and linear differential equations more generally have provided promising exponential speed-ups over classical computers on a set of problems with high real-world interest. However, extending this to a nonlinear problem has proven challenging, with exponential lower bounds having been demonstrated for the time scaling. We provide a quantum algorithm matching these bounds.  Specifically, we find that for a non-linear differential equation of the form $\frac{d\ket{u}}{dt} = A\ket{u} + B\ket{u}^{\otimes2}$ for evolution of time $T$, error tolerance $\epsilon$ and $c$ dependent on the strength of the nonlinearity, the number of queries to the differential operators   
    that approaches the scaling of the quantum lower bound of $e^{o(T\|B\|)}$ queries in the limit of strong non-linearity. Finally, we introduce a classical algorithm based on the Euler method allowing comparably scaling to the quantum algorithm in a restricted case, as well as a randomized classical algorithm based on path integration that acts as a true analogue to the quantum algorithm in that it scales comparably to the quantum algorithm in tailored cases where sign problems are absent.

\end{abstract}

\printglossaries

\section{Introduction}

The Hamiltonian Simulation Problem seeks to simulate quantum interactions with the use of a quantum computer; as was the first purpose proposed for the device when conceived by Richard Feynman, in 1982 \cite{1982IJTP...21..467F}). This is achieved by solving the following differential equation:

\begin{equation}
    i\frac{d\ket{u}}{dt} = H\ket{u},
\end{equation}
as defined by the Schrodinger equation describing the evolution of a quantum system; where $\ket{u}$ is a quantum state and $H$ is a Hermitian matrix. This may be analytically resolved to $\ket{u(t)} = e^{-iHt}\ket{u}$. The numerical approximation of such $\ket{u(t)}$ is however believed to be exponentially hard for any classical algorithm in its generality~\cite{aaronson2011computational}. The problem is, in fact, BQP-complete; meaning that any problem that a quantum computer can solve, with probability at least $\frac{2}{3}$, in time polynomial in the number of bits of input may be reduced to the Hamiltonian Simulation Problem. The existence of polynomial-time quantum algorithms for the Hamiltonian Simulation Problem is therefore of particular interest; further optimizations and improvements of the algorithm remain an active field of research \cite{Berry_2006} \cite{Low2019hamiltonian} \cite{PhysRevLett.118.010501} .

It is then natural to ask whether more general differential equations permit a quantum algorithm with an exponential asymptotic advantage over an adequately comparable classical algorithm. The problem of linear differential equations, removing the condition of our evolution remaining unitary, a.k.a:

\begin{equation}
        \frac{d\ket{u}}{dt} = A\ket{u},
\end{equation}
for some complex matrix $A$, has been studied extensively \cite{Berry_2017}\cite{Xin}. In a paper by Berry \cite{Berry_2014}, the problem was reduced to the quantum algorithm for solving linear systems of equations by Harrow, Hassidim and Lloyd \cite{Harrow_2009}, providing us with the first algorithm with exponential improvement over known classical algorithms. Subsequent works expanded the context for which quantum differential equations solvers could be applied; \cite{Krovi_2023} removed the restriction on the diagonalizability of the matrix, while \cite{Berry_2024} provided an algorithm capable of solving time-dependent equations. 

Removing the condition of linearity has, however, proven to be more difficult. Previous papers \cite{Joseph_2020} \cite{dodin2021applications} provided us with heuristic quantum algorithms tackling the problem, but these papers fall short of providing a proper implementation or a rigorous analysis of the running time of the algorithms. A first successful attempt at providing a concrete algorithm to tackle the Quantum Nonlinear Differential Equations problem in a restricted case was made by Liu et al. in 2021 \cite{Liu_2021}, using the Carleman Linearization technique that will be employed in this paper. However, the restrictions provided by the paper are bound to dissipative systems (a.k.a. with a loss of energy); this would exclude the initial Hamiltonian Simulation problem from our formalism. A subset of non-linear equations with no dissipative factor is further studied in \cite{wu2024quantumalgorithmsnonlineardynamics}, equally using Carleman Linearization. A different linearization technique, the Koopman-von Neumann linearization, was employed in \cite{tanaka2025polynomialtimequantumalgorithm} to solve a restricted class of norm-preserving nonlinear equations. 

A natural nonlinear extension of the Hamiltonian Simulation problem is the study of nonlinear differential equation with a unitary solution. Indeed, nonlinear differential equations with unitary solutions describe several physical phenomena \cite{nonlinschro}; the Gross-Pitaevskii equation, for example, is a nonlinear differential equation that may be used to model quantum systems of identical bosons. \cite{1961NCim...20..454G}. Nonetheless; it is not clear whether we may gain any substantial advantage by using quantum computing for these equations, relative to a classical model. Quantum computing is an inherently linear theory; the operations we may apply to our quantum states are linear maps. Deterministic classical numerical methods with better stability for nonlinear equations generally make use at multiple copies of the state at each iteration (example: higher-order Runge-Kutta method), which translates poorly to Quantum computing due to the no-cloning theorem.

We will first briefly state our results in Section \ref{sec:results}. Then, Section \ref{sec:nogo} consists of a proof of our no-go result that shows that sub-exponential scaling with time would allow us to violate state discrimination lower bounds. In
Section \ref{sec:algo}, we will discuss potential algorithmic solutions approaching these bounds.
Subsection \ref{sec:linearization} describes a linearization of equation \ref{equ:problem}. Subsection \ref{sec:quantum}will make use of our linearization to construct a quantum algorithm. Finally, in subsection \ref{sec:classical} we will provide a set of potential avenues to dequantize the previous algorithm.

\section{Results}
\label{sec:results}
The fact that quantum mechanics is a linear theory provides a major obstacle in the path of finding fast
quantum algorithms for solving non-linear differential equations. 
The need to generalize of the equations can be seen from the results of Lloyd et al. In \cite{2020arXiv201106571L}, the following set of equations is proposed,

\begin{equation}
   \frac{d\ket{u}}{dt} = f(\ket{u})\ket{u}
\end{equation}
where $f(u)$ is a $d \times d$ anti-hermitian matrix that is an order $m$ polynomial of $\ket{u}$ and $\bra{u}$. While \cite{2020arXiv201106571L} claims that a quantum algorithm polynomial in the evolution time $T$ is possible for this context - the arguments presented are not fully rigorous, and fail to abide by the lower bounds that may be established on the simulation of these equations.

Due to the fact that we require non-linear
operations in the above differential equation, we utilize tensor operators $(\cdot)^{\otimes p} : \mathbb{C}^{D} \mapsto \mathbb{C}^{Dp}$. These allow us to copy a state multiple times and express the evolution as an arbitrary polynomial of the elements of our state $\ket{u}$.  We use this notation to be able to express the differential equation in the form of rectangular matrices $R_1,R_2,\ldots R_k$,

\begin{equation}
\label{equ:problem}
    \frac{du}{dt} = R_1u + R_2u^{\otimes2}+... + R_ku^{\otimes k}
\end{equation}
such that the transformation defined by the equation describes a unitary evolution on the vector $u$. We note here that certain operations on complex vectors (such as the absolute values used in the Gross-Pitaevskii equation) may be represented as polynomials after representing our complex numbers as pairs of reals.  For example, the action of $i$ on the complex number $(a+ib)$ can be expressed as 
\begin{equation}
    \begin{bmatrix}
    a \\ b
\end{bmatrix} \mapsto \begin{bmatrix}
    0 & 1 \\
    -1 & 0
\end{bmatrix}\begin{bmatrix}
    a \\ b
\end{bmatrix}
\end{equation} 
and similarly complex conjugation takes the form of a Pauli-Z operation.

 Our first result, proved in Section~\ref{sec:nogo}, adapts the complexity-theoretic arguments from Childs and Young \cite{PhysRevA.93.022314} to the context of the equations that we are studying to obtain a lower bound on the query complexity of solving non-linear differential equations.

\begin{restatable}{lemma}{lemmanogo}
 \label{lemma:nogo2}
 Let $u: \mathbb{R} \mapsto \mathbb{C}^D$ for a positive integer $D$ and let 
 there exists a choice of matrices $R_1$ to $R_k$ be rectangular matrices such that $R_j \in {\mathbb C}^{D\times D^j}$ and a $d$-sparse time-dependent Hamiltonian $H:\mathbb{R} \mapsto \mathbb{C}^{D\times D}$ such that:
    
    $$i\frac{du}{dt} = H(t)u + R_1u + R_2u^{\otimes2}+... + R_ku^{\otimes k}$$ 
    describe a norm preserving evolution on $u$ such that for every $t$ there exists a unitary matrix $V$ such that $u(t) = V u(0) V^\dagger$
    . No algorithm can provide an approximate solution $\tilde{u}(t)$ such that $\|\tilde{u}(t)\tilde{u}^\dagger(t) - u(t) u(t)^\dagger\|_{\rm Tr}\le 2\epsilon$  for any constant $\epsilon$ with fewer than $e^{\Omega(t \times \sum_{j=1}^{k}\|R_j\|2^{-j/2}j)}$ queries to an oracle $O_u$ on $u(0)$.
\end{restatable}

We note that the proof of this lower bound applies to the family of equations studied by \cite{2020arXiv201106571L}.  Previous works such as  \cite{https://doi.org/10.48550/arxiv.0812.4423} provide an algorithm with exponential time scaling respecting the bounds of Lemma \ref{lemma:nogo2}; however, it is not shown that these methods saturate the bound. Furthermore, the algorithm also scales exponentially in $\frac{1}{\epsilon}$. Given that more recent Hamiltonian simulation results were able to provide logarithmic scaling with the inverse of the error \cite{7354428} \cite{Low2019hamiltonian}, this provides a clear avenue for improvement.  Deciding whether it is possible to find an algorithm that has a query complexity that scales logarithmically with the error for the nonlinear context and also matches the scaling matching provided in Lemma \ref{lemma:nogo2} remained an open problem. Our paper provides an algorithm capable of doing so.

\begin{restatable}{theorem}{thmquantum}
\label{theorem:1}
    Suppose we have for bounded matrices $A$ and $B$ an equation of the form:

    $$\frac{d\ket{u}}{dt} = A\ket{u} + B\ket{u}^{\otimes2}$$
    Where $A$ is $d_{\|A\|}-sparse$ and $B$ is $d_{\|B\|}-sparse$, and the dimension of u is bounded by: $\dim(\ket{u}) = 2^n$. Then there exists a quantum algorithm simulating $u(T)$ for a time $T>0$ within error $E$ with respect to the Euclidean norm for some arbitrarily small constant $\delta'$ with at most

    $$O\left(\frac{e^{2Tc(1+\delta')}(d_A^2\|A\| + d_B^2\|B\|)\log(1/E)n}{\log\log(1/E)}\right)  $$
    oracle queries, where  $c$ is in $O(\|B\|)$.
\end{restatable}
While the differential equation in Theorem~\ref{theorem:1} is focused on quadratic non-linearities this can easily be generalized to any polynomial case~\cite{https://doi.org/10.48550/arxiv.1711.02552}.

The relation between the run times of our quantum algorithm, and established deterministic classical numerical methods such as the Runge-Kutta method, is not entirely clear. While it is known that we may replicate any deterministic classical algorithm with a quantum algorithm using polynomial resources \cite{doi:10.1119/1.1463744}, it may not be excluded that a polynomially optimal quantum algorithm using our input model could run exponentially worse than deterministic classical numerical methods that use different oracles for their inputs. This apparent contradiction stems from the fact that our quantum algorithms have access to a weaker oracular model; deterministic algorithms generally can directly access the values of the vector $u$; while the quantum algorithm can only obtain a quantum state representing the input vector, more akin to a probability distribution. Thus; to create a classical equivalent to this problem, we propose an input model of a probability vector. Oracle queries to the input vector for this model are then random samples from our probability vector. To perform our operations on this large random vector, we will also proceed by a random sampling algorithm, making use of the well-known Path Integral Montecarlo method \cite{doi:10.1142/1170}. As is often the case for these types of algorithms, the runtime will depend on the variation in the output of our sampling - which may be very large, depending on the input to the problem. We will treat this variance as a separate variable - a more precise definition will be provided later.  

\begin{restatable}{theorem}{thmclassical}
\label{theorem:2}
   Suppose we have for bounded matrices $A$ and $B$ an equation of the form:

    $$\frac{d\ket{u}}{dt} = A\ket{u} + B\ket{u}^{\otimes2}$$
    
    Where $A$ and $B$ are $d_A$ and $d_B$-sparse respectively, $\ket{u}$ is a norm 1 vector in $\mathbb{C}^{2^n}$ and the assumption is made that the linear transformation on $\ket{u}$ remains unitary at all times $t$. Suppose that the path integral formulation for the Carleman block matrix of this problem has a variance of $\mathbb{V}(V)$. Then there exists a classical algorithm which outputs $v(j)$ such that:
        $$\|v(j)-\bra{u(0)}U(T)\ket{u(0)}\| \leq \epsilon$$
    for a time $T>0$ concerning the Euclidean norm, using at most

    $$O\left(\frac{\mathbb{V}(V)(d_A^2+d_B^2)e^{O(Tc(1+\delta))}}{\epsilon^2}\right)$$
    
    queries to our oracles, with our constants defined as previously.
\end{restatable}

We note here that the constant $c$ used is identical to the constant used in theorem \ref{theorem:1}. We note that if we allow linear time scaling in $n$, the logarithm of the dimension, time scaling matching the quantum case is demonstrated in the same section. This is of interest, as the lower bound arguments for Lemma \ref{lemma:nogo2} do not rely on dimension scaling. We may further note that the output produced by our quantum algorithm must also be sampled to obtain an estimate of the expectation calculated by our classical algorithm. This process loses the logarithmic error scaling, making it instead linear with the inverse of $\epsilon$, as shown in Corollary \ref{corollary:errorScaling}. Thus the classical algorithm only has quadratically worse scaling with the error. 

For a more restricted case of equations which we will name \emph{Geometrically Local Equations} (note: these are strictly distinct from the quantum notion of a geometrically local Hamiltonian), an algorithm polynomially matching the quantum algorithm is achievable, as proposed in the following theorem,

\begin{theorem}
Suppose we have for bounded matrices $A$ and $B$ an equation of the form:

    $$\frac{d\ket{u}}{dt} = A\ket{u} + B\ket{u}^{\otimes2},$$
    where $\ket{u}$ is a norm 1 vector in $\mathbb{C}^{2^n}$ and the assumption is made that $\ket{u(t)}$ remains a unit vector at all times $t$. Assume further that our equation is geometrically $k-local$ in $\mathbb{Z}^D$. Then there exists a classical algorithm which outputs $v(j)$ in $\mathbb{C}^{2^n}$ such that:
        $$\|v(j)-\ket{u(T)}\| \leq \epsilon$$
    for a time $T>0$ with respect to the Euclidean norm, using at most

    $$O\left(\frac{(kT)^{D+1}e^{2(D+1)(|A|+2|B|)T}}{(2\epsilon)^{D+2}}\right)$$   
    queries to classical oracles $O_A, O_B$ such that $O_A(i,j)$ yields the value of the $i^{\rm th}$ non-zero matrix element in row $j$ and similarly $O_B(i,j)$ yields the corresponding matrix element in row $j$ of $B$ and oracles $f_A,f_B$ such that $f_A(x)=y$ if and only if $A_{xy}$ is the $i^{\rm th}$ non-zero matrix element in row $x$ and $f_B$ yields the corresponding locations of the non-zero matrix elements of $B$.
\end{theorem}

\section{Quantum Lower Bound Argument}
\label{sec:nogo}

We will make use of the arguments from \cite{PhysRevA.93.022314} to create a lower bound for the problem we are studying. Our lower bound builds upon the results from section 2 in \cite{PhysRevA.93.022314}, which we summarize in the following lemma. We will require equations of the form:

\begin{equation}
\label{eq:nogo}
    i\frac{d}{dt}\ket{\psi} = H(t)\ket{\psi} + K\ket{\psi},
\end{equation}
where $K$ is a non-linear operator defined by: 

\begin{equation}
\label{eq:K}
    \bra{x}(K\ket{\psi}) = \kappa(|\braket{x|\psi}|)\braket{x|\psi},
\end{equation} 
for any computational basis vector $\ket{x}$, and some function $\kappa: \mathbb{R}^+ \to \mathbb{R}^+$. Thus, $K$ is a diagonal matrix with elements $\kappa(|\braket{x|\psi}|)$. $H(t)$ is a time-dependent hermitian operator chosen according to $K$. To avoid the quadratic term of $\kappa(|\braket{x|\psi}|)$ and allowing us to construct odd polynomials; we may also alternatively define, for $\braket{x|\psi} = a_x + b_xi$,

\begin{equation}
\label{eq:Kodd}
    \bra{x}(K\ket{\psi}) = \kappa(a+b)\braket{x|\psi},
\end{equation} 

We note here that for the restricted domain $b_x = 0$, $a_x \geq 0$ in which the arguments of \cite{PhysRevA.93.022314} are made, this alternative definition is equivalent. 
We will equally make use of the following definition from Equation 8 of \cite{PhysRevA.93.022314}:

\begin{equation}
\Bar{\kappa}(z) := \kappa \left(\left(\frac{1+z}{2}\right)^{1/2}\right) - \kappa \left(\left(\frac{1-z}{2}\right)^{1/2}\right).
\end{equation}

Then; in the Bloch sphere representation, our state $\ket{\phi} = (x,y,z)$ under the action of only the nonlinearity will act as:

\begin{equation}
\label{eq:rotation}
   \frac{d}{dt}(x,y,z) = \Bar{\kappa}(z)(-y,x,0).
\end{equation}

Then the following lemma applies:

\begin{lemma}[Section 2.3 of \cite{PhysRevA.93.022314}] \label{lemma:dist}
Suppose there are constants $g, \delta >0$ such that $\Bar{\kappa}(z) \geq gz$ for all $0 \leq z \leq \delta$. Then, evolving the equation \ref{eq:nogo} for time $O(\frac{1}{g}\log\left(\frac{1}{\epsilon}\right))$ will allow us to distinguish 2 chosen states $\ket{u}$ and $\ket{v}$ with overlap $|\langle{u}|{v}\rangle|\le 1-\epsilon$
with success probability greater than $\frac{4}{5}$.
\end{lemma}
 First, we must define the oracles we will use throughout this paper. 

\begin{definition}[Oracles $O_A$, $O_B$, $O_u$]
\label{def:oracles}
We define quantum oracles that provide us with access to the initial state $u(0)$ as well as the matrices $A$ and $B$ in the following manner. 
\begin{itemize}
    \item Given a matrix $A \in 	\mathbb{C}^{2^n \times 2^n}$ of sparsity $d_A$, integers $i \leq 2^n$, $j \leq d_A$. We will define the oracle $O_A\ket{i}\ket{j}\ket{0}\ket{0} = O_A\ket{i}\ket{j}\ket{p}\ket{q}$ where $p$ is an integer indicating the location of the $j-th$ non-zero column of row $i$ of $A$, and $q$  is a binary representation of $A[i,p]$.
    \item Given a matrix $B \in 	\mathbb{C}^{2^n \times 2^{2n}}$ of sparsity $d_B$, integers $i \leq 2^n$, $j \leq d_B$. We will define the oracle $O_B\ket{i}\ket{j}\ket{0}\ket{0} = O_B\ket{i}\ket{j}\ket{p}\ket{q}$ where $p$ is an integer indicating the location of the $j-th$ non-zero column of row $i$ of $B$, and $q$  is a binary representation of $B[i,p]$.
    \item Given an initial state vector $\ket{u(0)} \in \mathbb{C}^{2^n}$, we define $O_u$ as the unitary operator that prepares the initial state via $O_u\ket{0}^{\otimes  n} = \ket{u(0)}$.
\end{itemize}
\end{definition}
We will use lemma \ref{lemma:dist} and the above definition of $O_u$ to prove our main lemma of this section, which is a strengthening of the lower bounds contained in~\cite{PhysRevA.93.022314}.

\lemmanogo*

\begin{proof}
    
Let us make a choice of $R_j$ such that for all $\ket{x}$ and $\ket{u}$
\begin{equation}
    \bra{x}(R_j\ket{u}^{\otimes j}) = a_j|\braket{x|u}|^j\braket{x|u}
\end{equation} 
Then we have that:

\begin{equation}
    \lim_{z \to 0} \Bar{\kappa}(z) = \lim_{z \to 0}\frac{a_1\left(\frac{1+z}{2}\right)^{1/2} + a_2\left(\frac{1+z}{2}\right)^{2/2} + ... + a_k\left(\frac{1+z}{2}\right)^{k/2}}{z}.
\end{equation}
We make use of l'Hospital's rule to evaluate this limit, differentiating the nominator and denominator;
\begin{equation}
    \lim_{z \to 0} \Bar{\kappa}(z) = \lim_{z \to 0}\frac{\sum_{j=1}^k 2^{-j/2 -1}a_j(1+z)^{j/2-1}+2^{-j/2 -1}a_j(1-z)^{j/2-1}}{1} = \sum_{j=1}^k 2^{-j/2}a_j.
\end{equation}
Thus; as these are analytic functions, we may say that $g$ as defined in Lemma \ref{lemma:dist} may take any value greater than $\sum_{j=1}^k 2^{-j/2}a_j$; and thus, we may distinguish $\ket{\psi}, \ket{\phi}$ by simulating our differential equation for time $T = O\left(\frac{1}{\sum_{j=1}^k 2^{-j/2}a_j}\log\left(\frac{1}{\epsilon}\right)\right)$. We note here that for such a $\kappa$, equation \ref{eq:nogo} may be written in the form:

\begin{equation}
    \frac{du}{dt} = R_1u + R_2u^{\otimes2}+... + R_ku^{\otimes k},
\end{equation}
as required for Lemma \ref{lemma:nogo2}, with $\|R_{j+1}\| = a_j$.

Suppose we have access to an oracle to some unknown state, and are promised that the oracle either yields a state operator $\ketbra{\psi}{\psi}$ or $\ketbra{\phi}{\phi}$ such that $\|\ketbra{\psi}{\psi} - \ketbra{\phi}{\phi}\|_{\rm Tr}\ge 2\Delta$ for $\Delta>0$. By the Helstrom bound, $N = \theta\left(\frac{1}{\Delta}\right)$ copies of our state (or oracle queries to $O_u$) are necessary to determine with probability $\frac{2}{3}$ which state the oracle generates. (This is an older result, but is stated in a similar form in Section 4 of \cite{PhysRevA.93.022314}). The success rate of our algorithm for Lemma \ref{lemma:dist} depends on the distance between our states at the end of our evolution; we may permit an error $\epsilon$ sufficiently small to maintain a success rate greater than $2/3$ with triangle inequality. 

Suppose, now, that the differential equation can be simulated for time $t$ and error $\epsilon$ with $e^{o(t \times \sum_{j=2}^{k}a_j2^{-j/2}j)}$ queries to our oracle on our state $\ket{\psi}$ or $\ket{\phi}$. Then using our algorithm in Lemma \ref{lemma:dist}, the number $N$ of oracle queries required to distinguish $\ket{\psi}$from  $\ket{\phi}$ with success probability $\geq 2/3$  will be:

\begin{equation}
    N = e^{o(T \times \sum_{j=2}^{k}a_j2^{-j/2}j)} = e^{o\left(\left(\frac{1}{\sum_{j=1}^k 2^{-j/2}a_j}\log\left(\frac{1}{\epsilon}\right)\right) \times \sum_{j=2}^{k}a_j2^{-j/2}j\right)},
\end{equation}

\begin{equation}
    N=e^{o\left(\log\left(\frac{1}{\epsilon}\right)\right)} = o\left(\frac{1}{\epsilon}\right).
\end{equation}
This violates the Helstrom bound, and thus, we conclude that no such algorithm for solving our differential equation exists. This concludes our proof of Lemma \ref{lemma:nogo2}.
\end{proof}

We will also further note that a lower bound of $d\|A\|t$ will be inherited from the regular Hamiltonian Simulation problem; which is shown, for example, in \cite{Berry_2015}. 

\section{Quantum Algorithms for Non-Linear Differential Equations}

\label{sec:algo}
We will now attempt to find algorithms that approach our exponential lower bound from the previous section.

\label{sec:linearization}

As in \cite{Liu_2021} \cite{wu2024quantumalgorithmsnonlineardynamics}, we will start by making use of the Carleman Linearization process to obtain a higher-dimensional linear approximation of our equation. As a first step; we reduce our polynomial differential equation into the quadratic case. That this is permitted is a result taken directly from \cite{https://doi.org/10.48550/arxiv.1711.02552}. 

\begin{lemma}[Proposition 3.4 of \cite{https://doi.org/10.48550/arxiv.1711.02552}]
\label{lem:prop}
Let $u:\mathbb{R} \mapsto \mathbb{C}^{D\times D}$ and let $R_j\in \mathbb{C}^{D\times D^j}$.
   The $k$-th order system ($k  \geq 2$):

   $$\frac{du}{dt} = R_1u + R_2u^{\otimes2}+... + R_ku^{\otimes k}$$
  may be reduced to a quadratic system in $\Tilde{u}$, that is there exist matrices $A$ and $B$ such that
\begin{equation}\label{quadequation}
    \frac{d\Tilde{u}}{dt} = A\Tilde{u} + B\Tilde{u}^{\otimes2}
\end{equation}
  where $\Tilde{u} := \{u, u^{\otimes2}, u^{\otimes3},\ldots, u^{\otimes (k-1)}\}$. Moreover; $A$, $B$ respect:

  $$\|A\| \leq \max_{1\leq i \leq k-1} (k-i) \sum_{j=1}^i \|R_j\|$$
  $$\|B\| \leq (k-1) \sum_{j=2}^k \|R_j\|$$
here $\|\cdot\|$ refers to the spectral norm.
  
\end{lemma}
We will continue using $\|\cdot\|$ as a reference to the spectral norm throughout the paper.

To describe our matrices $A$ and $B$, it will first be necessary to define the \emph{transfer matrices}:

\begin{equation}
    T^i_{i+j} := \sum_{j = 0}^{m-1}I^{\otimes j} \otimes R_i \otimes I^{\otimes(k-j-1)}
\end{equation}

Our matrices $A$ and $B$ will be described as:

\begin{equation}
    A = 
    \begin{bmatrix}
        T^1_1 & T^1_2 & T^1_3 &\dots& T^1_{k-1} \\
        0 & T^2_2 & T^2_3 &\dots &T^2_{k-1} \\
        0 & 0 & T^3_3 &\dots &T^3_{k-1} \\
        \vdots & \vdots & \vdots & \ddots & \vdots \\
        0 & 0 & 0 & 0 &T^{k-1}_{k-1} 
        
    \end{bmatrix}
\end{equation}

\begin{equation}
    B = 
    \begin{bmatrix}
        T^1_k & 0 & 0 &\dots & 0 \\
        T^2_k & T^2_{k+1} & 0 &\dots & 0 \\
        T^3_k & T^3_{k+1} & T^3_{k+2} &\dots & 0 \\
        \vdots & \vdots & \vdots & \ddots & \vdots \\
        T^{k-1}_k & T^{k-1}_{k+1} & T^{k-1}_{k+2} &\dots &  T^{k-1}_{2(k-1)}
    \end{bmatrix}  
    \otimes \bra{k-1}
\end{equation}
We note that $\|T^i_{i+j}\|\leq j\|T_j\|$ hold by the triangle inequality and the subadditivity of the operator norm; our bounds on the norms of $A$ and $B$ follow. A further important consequence of this is that the sparsity remains manageable; $T^i_{i+j}$ has a sparsity of at most $jd_{R_i}$, where $d_{R_i}$ is the sparsity of $R_i$. We may then say the following: 

\begin{lemma}
    Given matrix $A$ and $B$ as defined in lemmma \ref{lem:prop}, the sparsities $d_A$ and $d_B$ may be bounded as:

    $$d_A \leq \max_{1\leq i \leq k-1} (k-i) \sum_{j=2}^{k} (j-i)d_{R_i}$$

    $$d_B \leq \max_{1\leq i \leq k-1} \sum_{j=1}^{i} (k-j)d_{R_i} $$
\end{lemma}

 The transformation of Lemma~\ref{lem:prop} is exact and will preserve unitarity in the higher dimensional space if $A$ and $B$ are anti-Hermitian. 
Let us then define:
\begin{equation}
    u_m = \Tilde{u}^{\otimes m}
\end{equation}
\begin{equation}
    A_m = \sum_{j = 0}^{m-1}I^{\otimes j} \otimes A \otimes I^{\otimes(k-j-1)}
\end{equation}

\begin{equation}
    B_m = \sum_{j = 0}^{m-1}I^{\otimes j} \otimes B \otimes I^{\otimes(k-j-1)}
\end{equation}
We note that $\|A_m\|\leq m\|A\|$ and $\|B_m\|\leq m\|B\|$ hold by the triangle inequality and the subadditivity of the operator norm.  
The following lemma describes our Carleman linearization, as shown in \cite{https://doi.org/10.48550/arxiv.1711.02552},

\begin{lemma}[Proposition 3.2 of \cite{https://doi.org/10.48550/arxiv.1711.02552}]\label{carleman}

If $\Tilde{u}$ solves equation \ref{quadequation}, then 

$$\frac{du_m}{dt} = A_m u_m + B_m u_{m+1}.$$
\end{lemma}
The equation of Lemma \ref{carleman} may be rewritten as the following infinite dimensional block matrix equation,

$$\begin{bmatrix} 
    \frac{du_{1}}{dt}  \\
    \frac{du_{2}}{dt}  \\
    \frac{du_{3}}{dt}  \\
    \vdots \\
    \end{bmatrix} =
    \begin{bmatrix} 
    A_{1} & B_{1} &0 &0&\dots \\
     0 & A_{2} & B_{2} &0& \dots \\
     0 &  0 & A_{3} & B_{3} & \dots \\
    \vdots &      \vdots  &\vdots & \vdots& \ddots
    \end{bmatrix}\begin{bmatrix} 
    u_{1}  \\
    u_{2}  \\
    u_{3}  \\
    \vdots \\
    \end{bmatrix}.$$
We define our solution vector to the above equation as $u_\infty$ and the infinite-dimensional block matrix as $G$. Unless otherwise specified; we will here on out treat each block as a unit, rather than referring to specific coordinates of the matrix; thus, $G_{2,3} = B_2$, for example. As shown in \cite{https://doi.org/10.48550/arxiv.1711.02552}, this solves for each $u_m$ exactly. It remains to find an appropriate approximation of this method for $u_1$, which corresponds to the solution of our initial nonlinear equation. We proceed by a truncation,

$$\begin{bmatrix} 
    \frac{du_{1}}{dt}  \\
    \frac{du_{2}}{dt}  \\
    \frac{du_{3}}{dt}  \\
    \vdots \\
    \frac{du_m-1}{dt}  \\
    \frac{du_{m}}{dt}  \\
    \end{bmatrix} \approx
    \begin{bmatrix} 
    A_{1} & B_{1} &0 &0&\dots &0&0\\
     0 & A_{2} & B_{2} &0& \dots &0&0\\
     0 &  0 & A_{3} & B_{3} & \dots &0&0\\
    \vdots &      \vdots  &\vdots & \vdots& \ddots &&\\
    0& 0& 0& 0& \dots& A_{m-1} & B_{m-1}\\
   0& 0& 0& 0& \dots& 0 & A_{m}\\
    \end{bmatrix}
    \begin{bmatrix} 
    u_{1}  \\
    u_{2}  \\
    u_{3}  \\
    \vdots \\
    u_{m-1}  \\
    u_{m}  \\
    \end{bmatrix}.$$
We will use $G_m$ to describe the above truncation of $G$. Theorem 4.2 in \cite{https://doi.org/10.48550/arxiv.1711.02552} provides us an estimate of our error for $0  < t < \frac{1}{\|B\|}$; however, their estimate diverges with $m$ for $t>\frac{1}{\|B\|}$. As is common with approximation techniques, it is, therefore, necessary for us to divide our evolution time into multiple smaller time steps. 

We note that this requires us to update not only $u_1$ at a given step; but also $u_2$ to $u_m$, since our linear system is dependent on these variables. To ensure that all variables are appropriately updated, a larger matrix is required at each additional time step we add. The starting $m$ must therefore be sufficiently large to cover the required truncations at all steps. We will compute the required truncations for each step and the required starting $m$ within this section.

For convenience, we enumerate our time steps as $j = 1$ to $j=J$, where $j=1$ is the last time step. This is because the size of the matrix we consider is also shrinking as we approach the last time step. We define the following vector:
\begin{equation}
    v(j, m) := \{v_1(t_j),v_2(t_j),\ldots, v_{m}(t_j))\}
\end{equation}
as our approximation of 
\begin{equation}
   u(j, m) := \{u_1(t_j), u_2(t_j), \ldots, u_{m}(t_j)\} 
\end{equation}
Let 

\begin{equation}
    v(j):= v(j, m_j)
\end{equation}
\begin{equation}
    u(j):= u(j, m_j)
\end{equation}
$m_j$ corresponds to the value of the $m$ used for our truncated block matrix $G_m$ at step $j$. (The value of $m_j$ will be chosen later.) The error $\epsilon_j$ on our vector $v(j)$ is defined as 

\begin{equation}
    \epsilon_j := |u(j)-v(j)|.
\end{equation}
Define further our time interval leading to step $j$:
\begin{equation}
  \hat{t}_{j} := t_{j-1}-t_j  
\end{equation}
Let $u[m]$ be the truncation of the block vector $u$ to blocks $1$ to $m$. Similarly, let $M[m_1,m_2]$ be the truncation of the block matrix $M$ restricted to block rows $1$ to $m_2$ and block columns $1$ to $m_2$. To denote instead the exclusion of rows/columns $1$ to $m$, write $\bar{m}$.

We may remark here that on the Lebesgue space with Euclidean norm $L^2$, for row-finite $M$, the initial value problem $\frac{dx}{dt} = Mx$, $x(0)$ has a unique solution $x(t) = e^{Mt}x(0) = \sum_{j \in \mathbb{N}}\frac{M^jt^j}{j!}x(0)$; further, in the case of matrices generated by Carleman Linearization, this solution will be finite everywhere \cite{carl} \cite{banach}. 

Thus, the exact solution at a given time step, $u(j)$ will be described as:
\begin{equation}
\label{eq:uj}
u(j) = \left(e^{\hat{t}_{j+1}G}u(j+1, \infty)\right)[m_j] = \left(e^{\hat{t}_{j+1}G}[m_j,\infty]u(j+1, \infty)\right)
\end{equation}
Our approximation for the first $m_j$ values will be 

\begin{equation}
\label{eq:vj}
    v(j) = \left(e^{\hat{t}_{j+1}G_{m_{j+1}}}v(j+1)\right)[m_j] = \left(e^{\hat{t}_{j+1}G_{m_{j+1}}}[m_j,m_{j+1}]v(j+1)\right)
\end{equation}In the following lemma; we will define the length of each time step and bound the error introduced at each time step by our truncation. The remainder of the section will serve to prove these bounds. 

\begin{lemma} \label{lemma:time}
    Let $u(j)$, $v(j)$, $m_j$ be defined as in \eqref{eq:uj},\eqref{eq:vj} where $G_{m_j}$ is the truncated Carleman matrix for step $j$. Let
$m_j = \sum_{s=1}^{j-1}k_j +1
$ where 
$k_j = \log_{1+\delta}\left(\frac{1}{\epsilon}\right)+\log_{1+\delta}(j)$ and let $
|v(j+1)-u(j+1)| \leq \epsilon_{j+1}$ for some $\epsilon>0$, $\delta >0$ and all $j \leq w$  . Then there exists a constant $c$ such that:

\begin{enumerate}
    \item The truncation error at time step $j$ may be bounded by: $$\Big|e^{c\hat{t}_jG_{j+1}}v(j+1)|(m_j) - u(j)|\Big| \leq \frac{\epsilon}{j\delta^2} + \epsilon_{j+1},$$
where 
$$ \hat{t}_j = \frac{1}{jc(1+\delta)}.$$
\item Our evolution at time step $t_j$ is $\frac{\epsilon}{j\delta^2}$ close to a unitary matrix with respect to the operator norm. 
\item Given a time evolution with $w$ timesteps, the total error of the truncation process $\epsilon_{tot} = |v(0) - u(0)|$ may then be bounded by :
$$\epsilon_{tot} \leq \frac{\ln(w)\epsilon}{\delta^2}.$$
\end{enumerate}

\end{lemma}

\begin{figure}
\centering
\begin{tikzpicture}
\node[draw] at (1,-1) {$F_j$};
\node[draw] at (3.5,-1) {$F_{j+1}[m_j,m_{j+1}]$};

\node[draw] at (7.2,-1) {$F[m_j,\bar{m}_{j+1}]$};

\node[draw] at (7.2,-5) {$F$};
\node[draw] at (3,-3) {$F_{j+1}$};

\node[draw] at (9.5,-1) {$v_j$};

\node[draw] at (9.5,-3) {$v_{j+1}$};

\draw [densely dotted](0.1,-0.1) -- (1.9,-0.1) -- (1.9,-1.9) -- (0.1,-1.9) -- (0.1,-0.1);

\draw (0,0) -- (6,0) -- (6,-2) -- (0,-2) -- (0,0);

\draw [dotted](6,0) -- (8, 0);

\draw [loosely dotted](0.1,-0.) -- (8, -8);

\draw [dotted](6,-2) -- (8, -2);

\draw [densely dotted](0,0) -- (6,0) -- (6,-6) -- (0,-6) -- (0,0);

\draw [densely dotted](0,0) -- (6,0) -- (6,-6) -- (0,-6) -- (0,0);

\draw [densely dotted](9.1,-0.1) -- (9.9,-0.1) -- (9.9,-1.9) -- (9.1,-1.9) -- (9.1,-0.1);

\draw (9,0) -- (10,0) -- (10,-6) -- (9,-6) -- (9,0);

\draw [densely dotted](10,-6) -- (10,-8);

\draw [densely dotted](9,-6) -- (9,-8);

\end{tikzpicture}

\caption{Illustration of 1 step of our process. $F_{j+1}$ restricted to the space of $m_j$  is applied onto $v_{j+1}$ to return $v_j$. } \label{fig:M1}
\end{figure}
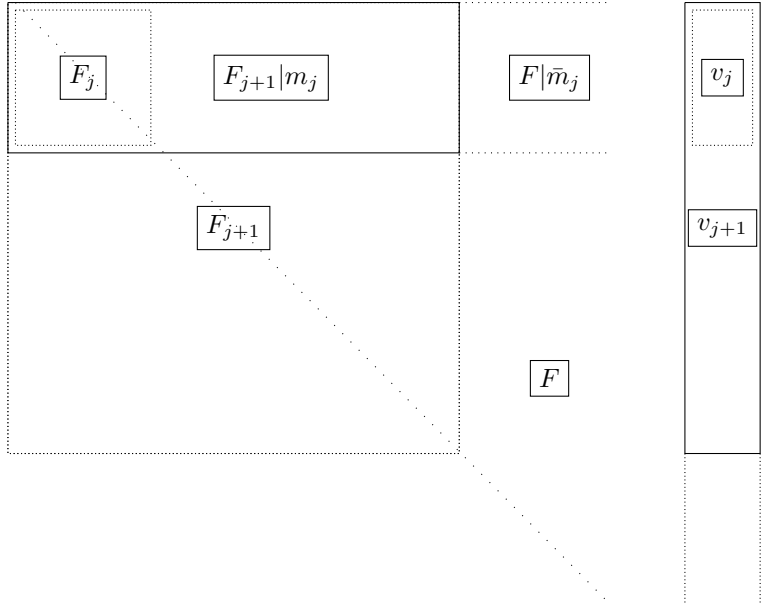

\begin{proof}

Let 

\begin{equation}
    F(t):= e^{tG},
\end{equation} and 

\begin{equation}\label{eq:fj}
    F_j(t):=e^{tG_j}.
\end{equation} We note that $F(t)$ is identical to $F_j(t)$ on all coordinates where $F_j(t)$ is defined, due to $G$ being block upper triangular, in particular:

\begin{equation}
    F_{j+1}(t)[m_j,m_{j+1}] = F(t)[m_j,m_{j+1}]
\end{equation}

Our error on $v_j$ is a combination of carrying forward the existing error on $v_{j+1}$ relative to $u_{j+1}$; and the tail of $F$ that is thrown away. (See  Figure  \ref{fig:M1}). $m_{j+1}$ is chosen to be sufficiently large to keep the latter error small. Thus we may bound our error vector as:
\begin{equation}
    \|v(j)-u(j)\| 
    = \|F_{j+1}(\hat{t}_j)[m_j,m_{j+1}]v(j+1) - F(\hat{t}_j)[m_j,\infty]u(j+1)\| 
\end{equation}
\begin{equation}
    =\|F_{j+1}(\hat{t}_j)[m_j,m_{j+1}]v(j+1) - (F(\hat{t}_j)[m_j,m_{j+1}]u(j+1)  + F(\hat{t}_j)[m_j,\bar{m}_{j+1}]u(j+1))\| 
\end{equation}
\begin{equation}
 \leq \|F(\hat{t}_j)[m_j,m_{j+1}]v(j+1) - F(\hat{t}_j)[m_j,m_{j+1}]u(j+1)\| + \|(F(\hat{t}_j)[m_j,\bar{m}_{j+1}]u(j+1)\|
\end{equation}
Now, as $F(\hat{t}_j)$ is a norm-preserving linear operator, the restriction $F(\hat{t}_j)[m_j,m_{j+1}]$ will be monotonically non-increasing. Further, $\|u(j)[\bar{m}_{j+1}]\| \leq \|u(j)\| = 1$. Thus we obtain:
\begin{equation}
\label{eq:ej}
    \|v(j)-u(j)\| 
\leq \epsilon_{j+1} + \|F(\hat{t}_j)[m_j,\bar{m}_{j+1}]\|.
\end{equation}
We note further that $F(\hat{t}_j)[m_j,\bar{m}_{j+1}]$ corresponds to the distance of our evolution from a unitary - as the exact evolution is unitary. Let us define $k_j$, representing the additional copies in our block vector required from the previous step $j+1$ to approximate our vector at step $j$:

\begin{equation}
    k_j := m_{j+1}-m_j.
\end{equation}
 As previously, we treat $F$ as a block matrix, acting on $\{v_1, v_2, v_3, ...\}$; $F_{l,k}$ will refer to the block at coordinates $(l,k)$ in this block matrix. 
 The tail of our truncation, $F(\hat{t}_j)[m_j,\bar{m}_{j+1}]$, may be bounded as:

\begin{equation}
\label{eq:tail}
    \big\|F(\hat{t}_j)[m_j,\bar{m}_{j+1}]\big\| \leq \sum_{m = 1}^{m_j}\sum_{s = k}^{\infty} \|F_{m,m_{j}+s}(\hat{t}_j)\|.
\end{equation}

We note here that $m_{j}+k = m_{j+1}$. We will show that $\|F_{m_j,m_j+k}(\hat{t}_j)\|$ decays exponentially in $k$ and $(m_j-m)$ for the correct choice of $m_j$, $\hat{t}_j$ and $k_j$, thus allowing us to bound this summation as a geometric series.

We define a new matrix $\hat{G}$ as : 

\begin{equation}
    \hat{G}_{i,j} := \|G_{i,j}\|
\end{equation}

and define $\hat{F}$ as $e^{\hat{G}}$. An important property here is that: 

\begin{equation}
\label{fhat}
    \hat{F}_{i,j} \geq \|F_{i,j}\|
\end{equation}

This is a direct consequence of the sub-additive and sub-multiplicative properties of the operator norm: $\|A+B\| \leq \|A\|+ \|B\|$, $\|AB\| \leq \|A\|\|B\|$. $F_{i,j}$ will be expressed as:

\begin{equation}
    F_{i,j}
 = \left(\sum_{k = 0}^\infty \frac{G^k}{k!}\right)_{i,j} = \left(\sum_{k = 0}^\infty \frac{(G^k)_{i,j}}{k!}\right)
\end{equation}

\begin{equation}
\|(G^k)_{i,j}\|
 = \|\sum_{a_1 \in \mathbb{N}}\sum_{a_1 \in \mathbb{N}}... \sum_{a_{k-1} \in \mathbb{N}}G_{i,a_1}G_{a_1,a_2}...G_{a_{k-2},a_{k-1}}G_{a_{k-1},j}\|
\end{equation}

\begin{equation}
\|(G^k)_{i,j}\|
\leq \sum_{a_1 \in \mathbb{N}}\sum_{a_1 \in \mathbb{N}}... \sum_{a_{k-1} \in \mathbb{N}}\|G_{i,a_1}\|\|G_{a_1,a_2}\|...\|G_{a_{k-2},a_{k-1}}\|\|G_{a_{k-1},j}\| = (\hat{G}^k)_{i,j}
\end{equation}

\begin{equation}
    \|F_{i,j}\|
 = \|\sum_{k = 0}^\infty \frac{(G^k)_{i,j}}{k!}\| \leq \sum_{k = 0}^\infty \frac{\|(G^k)_{i,j}\|}{k!} \leq \sum_{k = 0}^\infty \frac{(\hat{G}^k)_{i,j}}{k!} = \hat{F}_{i,j}.
\end{equation}
Using $\|A_m\|\leq m\|A\|$, $\|B_m\|\leq m\|B\|$, we may diagonalize our matrix $\hat{G}$ as:

\begin{equation}
    \hat{G} = \begin{bmatrix} 
    \|A\| & \|B\| &0 &0&\dots \\
     0 & 2\|A\| & 2\|B\| &0& \dots \\
     0 &  0 & 3\|A\| & 3\|B\| & \dots \\
    \vdots &      \vdots  &\vdots & \vdots& \ddots
    \end{bmatrix} = SJS^{-1},
\end{equation}

    \begin{equation}
      \hat{G} = \begin{bmatrix} 
    1 & \frac{\|B\|}{\|A\|} &\frac{\|B\|^2}{\|A\|^2} &\frac{\|B\|^3}{\|A\|^3}&\dots \\
     0 & 1 & \frac{2\|B\|}{\|A\|} &\frac{3\|B\|^2}{\|A\|^2}& \dots \\
     0 &  0 & 1 & \frac{3\|B\|}{\|A\|} & \dots \\
    \vdots &      \vdots  &\vdots & \vdots& \ddots
    \end{bmatrix}\begin{bmatrix} 
    \|A\| & 0 &0 &0&\dots \\
     0 & 2\|A\| & 0 &0& \dots \\
     0 &  0 & 3\|A\| & 0 & \dots \\
    \vdots &      \vdots  &\vdots & \vdots& \ddots
    \end{bmatrix}\begin{bmatrix} 
    1 & \frac{-\|B\|}{\|A\|} &\frac{\|B\|^2}{\|A\|^2} &\frac{-\|B\|^3}{\|A\|^3}&\dots \\
     0 & 1 & \frac{-2\|B\|}{\|A\|} &\frac{3\|B\|^2}{\|A\|^2}& \dots \\
     0 &  0 & 1 & \frac{-3\|B\|}{\|A\|} & \dots \\
    \vdots &      \vdots  &\vdots & \vdots& \ddots
    \end{bmatrix}  
    \end{equation}
    Where $S[m,m+k] = {m+ k -1 \choose k}\left(\frac{\|B\|}{\|A\|}\right)^k$, $J[m,m] = m\|A\|$ and $S^{-1}[m,m+k] = {m+ k -1 \choose k}\left(\frac{-\|B\|}{\|A\|}\right)^k$. Inserting the constant factor of $t$, we obtain the same $S, S^{-1}$ and $J[m,m] = tm\|A\|$.

Let's now demonstrate this diagonalization. It is clear that the eigenvalues of our infinite matrix $\hat{G}$ are $k\|A\|$ for all strictly positive integers $k$; we compute our corresponding eigenvectors from this. For some fixed eigenvalue $k\|A\|$,

\begin{equation}
    j\|A\|v_k[j]+ j\|B\|v_k[j+1] = k\|A\|v_k[j]
\end{equation}

\begin{equation}
    v_k[j+1] = \frac{(k-j)\|A\|v_k[j]}{j\|B\|}
\end{equation}
Setting $v_k[1]$ to $\frac{\|B\|^{k-1}}{\|A\|^{k-1}}$:

\begin{equation}
    \forall j \leq k: v_k[j] = {k \choose k-j}\frac{\|B\|^{k-j}}{\|A\|^{k-j}}
\end{equation}

\begin{equation}
    \forall j > k: v_k[j] = 0
\end{equation}
Which gives us the desired diagonalization. 

Given some diagonalization $SJS^{-1}$, we know that $e^{SJS^{-1}} = Se^JS^{-1}$. Making use of this fact, we may then compute the non-zero entries of $\hat{F}(t) = e^{\hat{G}t}$ as follows,

\begin{equation}
\hat{F}_{m,m+k}(t) \leq \sum_{j=m}^{m+k} {m+(j-m)-1\choose j-m}\frac{\|B\|^{j-m}}{\|A\|^{j-m}}e^{jt\|A\|}(-1)^{j+m+k-1}{m+k-1 \choose m+k -j}\frac{\|B\|^{m+k-j}}{\|A\|^{m+k-j}}
\end{equation}
\begin{equation}
    \hat{F}_{m,m+k}(t) \leq {k+m-1 \choose k}\frac{\|B\|^k}{\|A\|^k}\sum_{j=m}^{m+k} {k \choose j-m}e^{jt\|A\|}(-1)^{j+m+k-1}
\end{equation}
\begin{equation}
    \hat{F}_{m,m+k}(t) \leq {k+m-1 \choose k}\frac{\|B\|^k}{\|A\|^k}e^{mt\|A\|}\sum_{j=0}^{k} {k \choose j}e^{jt\|A\|}(-1)^{k-j}
\end{equation}
Using here, the binomial theorem:
\begin{equation}
\label{eq:bound}
    \hat{F}_{m,m+k}(t) \leq {k+m-1 \choose k}\frac{\|B\|^k}{\|A\|^k}e^{mt\|A\|}(e^{\|A\|t}-1)^k
\end{equation}
The binomial coefficients that appear above can be upper bounded by elementary functions via:

\begin{equation}
\label{eq:app}
    {k+m-1 \choose k} \leq \left(\frac{e(m+k-1)}{k}\right)^k
\end{equation}
We establish here that $\hat{t}_j \leq \frac{1}{\|A\|}$ for all j (we will force this later). Then:

\begin{equation}
    (e^{\|A\|\hat{t}_j}-1)^k \leq (\|A\|\hat{t}_j \times (e-1))^k
\end{equation}
We will further have that \begin{equation}
    m_j \times \hat{t}_j\|A\| \leq k_j
\end{equation}
Thus: 
\begin{equation}
    e^{m_j\|A\|\hat{t}_j} \leq e^{k_j}.
\end{equation}
Set the constant $c$ such that $c = \max\{\|A\|, e^{2}(e-1)\|B\|\}$. (Note here that $e^{2}(e-1) < 12.7$.) Then:

\begin{equation}
\label{eq:exp}
    \|F_{m,m+k_j}(\hat{t}_j)\| \leq \hat{F}_{m,m+k_j}(\hat{t}_j) \leq \left(\frac{\hat{t}_jc(m_j+k_j-1)}{k_j}\right)^{k_j}.
\end{equation}
Further, for $k \geq k_j$ we have that

\begin{equation}
\label{eq:exp2}
    \|F_{m,m+k}(\hat{t}_j)\| \leq \left(\frac{\hat{t}_jc(m_j+k-1)}{k}\right)^{k}.
\end{equation}
 holds as well.

For a given step $j$; we will now define

\begin{equation}
\label{eq:kj}
 k_j = \log_{1+\delta}\left(\frac{1}{\epsilon}\right)+\log_{1+\delta}(j), 
\end{equation}
for some chosen $\delta$. We remember that $m_j$ represents the number of elements in our block vector $v(j)$ at step $j$; our final block vector at step 1 being of size 1 (we remember here that the index $j$ runs in the opposite direction of our time, for convenience) and thus,

\begin{equation}
    m_j = 1 + \sum_{s=1}^{j-1} k_j = \sum_{s=1}^{j-1} k_j\log_{1+\delta}\left(\frac{1}{\epsilon}\right)+\log_{1+\delta}(j),
\end{equation}
\begin{equation}
\label{eq:mj}
    m_j\leq j(\log_{1+\delta}\left(\frac{1}{\epsilon}\right)+\log_{1+\delta}(j)).
\end{equation}
We then chose the size of our time step; defined as previously $\hat{t}_j = t_{j}-t_{j+1}$.

\begin{equation}
\label{eq:tj}
    \hat{t}_j = \frac{1}{jc(1+\delta)}.
\end{equation}
This respects all the bounds required previously. Find the inequality

\begin{equation}
    \|F_{m_j, m_j+k_j}(\hat{t}_j)\| \leq \left(\frac{1}{1+\delta}\right)^{\log_{1+\delta}\left(\frac{1}{\epsilon}\right)+\log_{1+\delta}(j)}= \frac{\epsilon}{j}.
\end{equation}
From equations \ref{eq:exp2} and \ref{eq:kj}, we may further obtain the bound:

\begin{equation}
\label{eq:comb1}
    \|F_{m_j, m_j+k_j+r}(\hat{t}_j)\| \leq \left(\frac{1}{1+\delta}\right)^r\frac{\epsilon}{j}.
\end{equation}
Additionally:

\begin{equation}
\label{eq:comb2}
    \|F_{m_{j}-s, m_{j}-s+k_{j}}(\hat{t}_j)\|\leq \left(\frac{\frac{1}{jc(1-\delta)}c(m_{j}-s+k_j-1}{k_j}\right)^{k_j} \leq \|F_{m_j, m_j+k_j}(\hat{t}_j)\|.
\end{equation}
From equation \ref{eq:comb1} and equation \ref{eq:comb2}, we obtain: 

\begin{equation}
    \sum_{m = 1}^{m_j}\sum_{s = k_j}^{\infty} \|F_{m,m_j+s}(\hat{t}_j)\| \leq
\sum_{m = 1}^{m_j}\sum_{s = k}^{\infty} 
\left(\frac{1}{1+\delta}\right)^{m_j-m}
\left(\frac{1}{1+\delta}\right)^{s-k+1}\frac{\epsilon}{j} \leq  \frac{\epsilon}{j\delta^2}.
\end{equation}
As this represents our distance from the exact operator, which is known to be unitary, our evolution matrix has a distance of at most $\frac{\epsilon}{j\delta^2}$ from a unitary. Combining the previous equation with equations \ref{eq:ej} and \ref{eq:tail} we may obtain:

\begin{equation}
    |u(j)-v(j)| \leq \epsilon_{j+1} + \frac{\epsilon}{j\delta^2}.
\end{equation}

We may use this equation to describe our error over the evolution. Suppose that we require $w$ steps to simulate our entire evolution. Setting $\epsilon_{N+1}$ to $0$ (as this is our initial state), the total truncation error $\epsilon_{tot}$ may be described as follows:

\begin{equation}
    \epsilon_{tot} \leq \sum_{j = 1}^{w} \frac{\epsilon}{j\delta^2} \leq \frac{\ln(w)\epsilon}{\delta^2}.
\end{equation}

\end{proof}

We now have established the size of our time steps for our linearization process and the error introduced by our truncation at each time step. It remains to simulate these linear equations.

\subsection{Proof of Theorem \ref{theorem:1}}
\label{sec:quantum}

We wish to simulate our equation for some time $T$. Using our bounds from Lemma \ref{lemma:time}, we may evolve our vector by $\hat{t}_j = \frac{1}{jc(1+\delta)}$ at time step $j$. Thus, to simulate our entire evolution, we require $w$ time steps, such that

\begin{equation}
    T \leq \sum_{j=1}^w \frac{1}{jc(1+\delta)}.
\end{equation}
But we know that

\begin{equation}
    \sum_{j=1}^w \frac{1}{jc(1+\delta)} \geq \frac{\log(w)}{c(1+\delta)}.
\end{equation}
Thus, a sufficient condition is that,

\begin{equation}
\label{eq:w}
   w \geq e^{Tc(1+\delta)}.
\end{equation}
We will simulate our evolution using a Linear Combination of Unitaries (LCU) method; applied upon a Taylor expansion of $e^{Gt}$. To evaluate our running time, it is first necessary to ensure that the norm of the matrices that we are simulating remains small. This requires us to splice our time steps further for our previously fixed matrices. We note that this time splicing does \emph{not} require us to expand our matrix further; as our matrix at each time step has already been determined to be sufficiently accurate for that time step.

\begin{lemma}
   Define: $\tau_j = \frac{\hat{t}_j}{2(\log_{1+\delta}\left(\frac{1}{\epsilon}\right)+\log_{1+\delta}(j))}.$   
Then for any given $F_j$ as defined in equation \ref{eq:fj}, the following inequality holds:
   
   $$\|F_j(\tau_j)\| \leq e.$$ 
\end{lemma}

\begin{proof}
   We note here that

   \begin{equation}
      \|F_j(\tau_j)\| = \|G_j(\tau_j)\|
   \end{equation}
 \begin{equation}
      \|F_j(\tau_j)\| \leq \left\|\exp\left(\tau_j\begin{bmatrix} 
    \|A\| & \|B\| &0 &0&\dots &0 \\
     0 & 2\|A\| & 2\|B\| &0& \dots &0 \\
     0 &  0 & 3\|A\| & 3\|B\| & \dots &0 \\ 
    \vdots &      \vdots  &\vdots & \vdots& \ddots & \vdots\\ 0& 0 & 0 & 0 & \dots &m_j\|A\|
    \end{bmatrix}\right)\right\|
   \end{equation}
Then, as our last matrix is strictly positive, we may bound the norm of the matrix exponential by the norm of the exponential of the norm of the largest row,  thus,

\begin{equation}
    \|F_j(\tau_j)\| \leq e^{\tau_j m_j (\|A\|+\|B\|)}.
\end{equation}
As $\hat{t}_j \leq \frac{1}{jc}$, $m_j \leq j(\log_{1+\delta}\left(\frac{1}{\epsilon}\right)+\log_{1+\delta}(j))$ and $c \leq \frac{\|A\| + \|B\|}{2}$, the conclusion follows.
\end{proof}

We note (see  Figure \ref{fig:M2}) that our matrix $F_{j+1}[m_j,m_{j+1}]$ may be represented as

\begin{equation}
    e^{G_{j+1}\tau_j}-e^{(G_{j+1}-G_j)\tau_j} + I.
\end{equation}

\begin{figure}
\centering
\begin{tikzpicture}
    \draw (0,0) -- (7,0) -- (7,-2) -- (2,-2) -- (0,0);

    \draw [dotted](-0.22,0.1) -- (7.1,0.1) -- (7.1,-7.25) -- (-0.2,0.1);

    \draw [dotted](2.1,-2.1) -- (7,-2.1) -- (7, -7) -- (2.1,-2.1) ;

    \node[draw] at (4,-1) {$F_{j+1}[m_j,m_{j+1}]$};

    \node[draw] at (5.5,-3.5) {$e^{(G_{j+1}-G_j)\tau_j}$};
    \node[draw] at (2,-4) {$e^{G_{j+1}\tau_j}$};

\end{tikzpicture}
\caption{Computing $F_{j+1}[m_j,m_{j+1}]$ from $e^{G_{j+1}\tau_j}$ and $e^{(G_{j+1}-G_j)\tau_j}$} \label{fig:M2}
\end{figure}
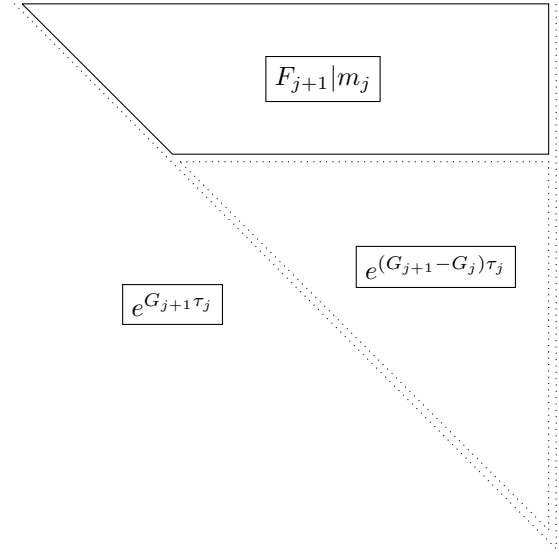
We may then describe our evolution matrix as: 

\begin{equation}
    \sum_{k=0}^\infty \frac{(G_{j+1}^k-(G_{j+1}-G_j)^k )\tau_j^k}{k!} + I = \sum_{k=0}^K \frac{(G_{j+1}^k-(G_{j+1}-G_j)^k)\tau_j^k}{k!}+I + R_K
\end{equation}
Next, note that $\|G_{j+1}\|\le \max_j \|G_j\|$ and by triangle ,inequality. $\|G_{j+1} - G_j\| \le 2\max_j \|G_j\|$.  We then have
\begin{equation}
 \left\|R_k\right\| \le \sum_{k=K+1}^\infty \frac{(3\max_j\|G_j\|\tau_j)^k}{k!} \in O\left(\frac{(\max_j \|G_j\| \tau_j)^K}{K!} \right).
\end{equation}
After the use of Stirling's inequality, we obtain that
\begin{equation}
 \left\|R_k\right\|  \in O\left(\frac{(e\max_j \|G_j\| \tau_j)^K}{K^K} \right).
\end{equation}
Solving for $\|R_k\|=\epsilon^*$ for some desired error target $\epsilon^*$ yields a Lambert W function whose asymptotic form then reveals that it suffices to choose a value of $K$ such that
\begin{equation}
    K = O\left(\frac{\log(\max_j \|G_j\| \tau_j/\epsilon^*)}{\log\log(\max_j \|G_j\| \tau_j/\epsilon^*)}\right),
\end{equation}
to obtain a $\epsilon^*$ approximation of our matrix.

Next we will represent our $G_j$ and $G_{j+1}-G_j$ as sums of unitaries of the form:
$\sum_{l=0}^L\alpha_l U_l$. We note that $A$ and $B$ may be decomposed into $O(d_A^2|A|\log N)$ unitaries and $O(d_B^2|B| \log N)$ unitaries respectively. This is implemented using edge coloring edge-coloring (where here $N = 2^n$ is our dimension; see \cite{Berry_2006}), and each unitary may be implemented using $O(1)$ copies of $O_A$ or $O_B$. Then, as it's easy to construct $G_j$ and $G_{j+1}-G_j$ from $A$ and $B$ using our definition, we can see that the $L$ we require may be bounded by $O(m_j(d_A^2|A| + d_B^2|B|))$. Now; the sum of unitaries may be implemented using the well-known Linear Combination of Unitaries (LCU) results; we will use here a robust version from \cite{7354428}:

\begin{lemma}[Lemma 4 of \cite{7354428}]
\label{lemma:LCU}
    Let $\Vec{U} = (U_1, . . . , U_M )$ be unitary operations and let $V = \sum_{m=1}^Ma_mU_m$ be $\Bar{\delta}$-close to a unitary. Let $a := \sum_{m =1}^M |a_m|$. We can approximate $V$ to within $O(\Bar{\delta})$ using $O(a)$ select $U$ and select $\Bar{U}$ operations.
\end{lemma}
We may now proceed with our proof of Theorem \ref{theorem:1}, which we restate below for convenience.
 
\thmquantum*
\begin{proof}[Proof of Theorem~\ref{theorem:1}]

We note here that our select operations may be implemented using $O(M)$ calls to our oracles. The total evolution cost for our LCU is then:

\begin{equation}
\label{equ:LCU}
    O(M\times a) = O(LK\times 1) = O\left(\frac{m_j(d_A^2\|A\| + d_B^2\|B\|)\log(1/\epsilon^*)n}{\log\log(1/\epsilon^*)}\right).
\end{equation}
It remains to establish our choices of $\epsilon^*$ produced by the approximate simulation of our $F_{j+1}|m_j$, as well as the $\epsilon_j$ distance between the simulated quasi-unitaries and the nonlinear dynamics, to ensure that our total error remains small throughout the evolution. 

Define a new variable $v_{sim}(j)$ to represent our simulated value of $v_j$ at time $t_j$.
Within a given $\hat{t}_j$, we have $2(\log_{1+\delta}\left(\frac{1}{\epsilon}\right)+\log_{1+\delta}(j))$ substeps. Let

\begin{equation}
    \label{equ:star}
    \epsilon^* = \frac{\epsilon}{j\delta^2}.
\end{equation}
By triangle inequality, we are then $\frac{2\epsilon}{j\delta^2}$ close to a unitary. Lemma \ref{lemma:LCU} will then allow us to simulate approximation of $F_{j+1}|m_j$ within $O(\frac{\epsilon}{j\delta^2})$. Given that our error accumulates approximately linearly when small, our simulation will have that

\begin{equation}
    |v_{sim}(j)-u(j)| = O\left(\frac{(\log_{1+\delta}\left(\frac{1}{\epsilon}\right)+\log_{1+\delta}(j))\epsilon}{j\delta^2}\right) + |v_{sim}(j+1)-u(j+1)|.
\end{equation}
The error from the entire evolution may then be described as

\begin{equation}
    |v_{sim}(0)-u(0)| = O\left(\sum_{j=1}^w\frac{(\log_{1+\delta}\left(\frac{1}{\epsilon}\right)+\log_{1+\delta}(j))\epsilon}{j\delta^2}\right),
\end{equation}

\begin{equation}
    |v_{sim}(0)-u(0)| = O\left(\frac{T^2\epsilon}{\delta^2}\right).
\end{equation}
For a desired error of $E$, it is then sufficient to set

\begin{equation}
    \epsilon = O\left(\frac{E\delta^2}{T^2}\right).
\end{equation}
Our total evolution cost will be over all values of $j$. 
Taking our value of $w$, the cost of step $j$ from equation \ref{equ:LCU} and \ref{equ:star}; we obtain that the total number of queries required by our algorithm is

\begin{equation}
    O\left(\sum_{j=1}^{w = e^{Tc(1+\delta)}}\frac{j\log_{1+\delta}(\frac{jT^2}{E\delta^2})(d_A^2\|A\| + d_B^2\|B\|)\log(T^2j/E)n}{\log\log(T^2j/E)}\right)
\end{equation}

\begin{equation}
  =O\left(\frac{e^{2Tc(1+\delta')}(d_A^2\|A\| + d_B^2\|B\|)\log(1/E)n}{\log\log(1/E)}\right), 
\end{equation}
where 
\begin{equation}
    c \leq  \max\{\|A\|, 12.7\|B\|\}, \label{eq:cbd}
\end{equation} and $\delta'> \delta$ remains an arbitrarily small constant. This concludes our proof of Theorem \ref{theorem:1}. \end{proof}

We have provided a quantum algorithm for our initial problem of solving a nonlinear differential equation as described in equation \ref{equ:problem}; with a scaling in time and in the norm of our nonlinear term that is similar to our lower bound from Lemma \ref{lemma:nogo2}, and logarithmic in the inverse of the error. 
We note however that what we obtain is not equivalent to a classical vector from which we may directly read values. For a fair comparison with our classical algorithm, we will thus compute the expected value of our output vector. 

\begin{corollary}
\label{corollary:errorScaling}
     Suppose we have an equation of the form

    $$\frac{d\ket{u}}{dt} = A\ket{u} + B\ket{u}^{\otimes2}$$
    where $A\in \mathbb{C}^{2^n\times 2^n}$, $B \in \mathbb{C}^{2^{n}\times 2^{2n}}$ are $d-sparse$ complex matrices. Let $U$ be some unitary matrix of matching dimension which may be queried in time $O(1)$. Then there exists a quantum algorithm computing the expectation value $\bra{u(T)}U\ket{u(T)}$ with probability greater than or equal to $2/3$  for a time $T>0$ within error $\epsilon$ using at most

    $$\Tilde{O}\left(\frac{e^{Tc}d^2\log(\frac{1}{\epsilon})}{\epsilon}\right)$$
    oracle queries where $c \in O(\|A\|+\|B\|)$ from~\eqref{eq:cbd}.
\end{corollary}
\begin{proof}
 We may assume without loss of generality that our inputs are real-valued and that we operate strictly within the reals (see the introduction for the necessary mapping). We will make use of the Hadamard test to obtain an estimate of $\bra{u(T)}U\ket{u(T)}$. We add an qubit $\ket{0}$, then the circuit
$$
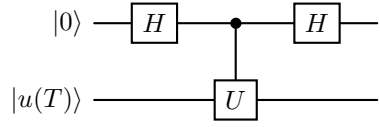
\begin{quantikz}
\lstick{$\ket{0}$} &\gate{H}& \ctrl{1} &\gate{H}&\\
\lstick{$\ket{u(T)}$} && \gate{U} &&
\end{quantikz}
$$
 creates the state 

\begin{equation}
    \ket{0}\ket{u}\to \frac{1}{\sqrt{2}}(\ket{0}\ket{u}+\ket{1}\ket{u}) \to (\ket{0}\ket{u}+\ket{1}U\ket{u}) \to \frac{1}{2}(\ket{0}(\ket{u}+U\ket{u})+\ket{1}(\ket{u}-U\ket{u}))
\end{equation}
The probability of measuring the first qubit to be $0$ is

\begin{align}
    \frac{(\ket{u}+U\ket{u})^\dagger(\ket{u}+U\ket{u})}{4} &= \frac{\bra{u}\ket{u}+\bra{u}U\ket{u}+\bra{u}U^\dagger\ket{u} + \bra{u}UU^\dagger\ket{u}}{4}\\
    &=\frac{1+(\bra{u}U\ket{u})^* + (\bra{u}U\ket{u})+ 1}{4}\\
    &=\frac{1+{\rm Re}(\bra{u}U\ket{u})}{2}.
\end{align}
The above Hadamard test circuit can be modified to yield the expectation value of the imaginary part can be estimated by simply applying a $S^\dagger$ gate to the control qubit before $U$ is applied.
 Given that $\bra{u}U\ket{u}$ will be strictly real, this is equivalent to $\frac{1}{2}(1+\bra{u}U\ket{u})$. We may employ the \emph{amplitude estimation} heuristic from \cite{Brassard_2002} find $a$ such that $|a -\frac{1}{2}(1+\bra{u}U\ket{u})| \leq \frac{\epsilon'}{2}$ with probability of failure less than $1/6$.  Repeating this for the imaginary part yields from the union bound a probability of failure of at most $1/3$.  These two processes combined require using at most $O(\frac{1}{\epsilon}')$ iterations of our circuit. Multiplying our result by 2 and removing 1 then gives us our desired estimate of $\bra{u(T)}U\ket{u(T)}$ within error $\epsilon'$. 

The amplitude estimation process described above requires a reflection about the state $\ket{u(T)}$.  This reflection can be implemented using at most a single query to an oracle that prepares $\ket{u(T)}$ from $\ket{0}$ and the inverse of this oracle.  The result of Theorem \ref{theorem:1} provides us with an approximate necessary oracle to $\ket{u'(T)}$ such that $\|\ket{u'(T)}-\ket{u(T)}\| \leq \epsilon''$ using at most

\begin{equation}
    N_{\rm queries} \in O\left(\frac{e^{2Tc(1+\delta')}(d_A^2\|A\| + d_B^2\|B\|)\log(1/\epsilon'')n}{\log\log(1/\epsilon'')}\right)  
\end{equation}

queries to our oracles $O_u$, $O_A$ and $O_B$. Setting $\epsilon' = \epsilon/2, \epsilon'' = \epsilon/2$ our result follows.

\end{proof}

The optimal scaling that can be achieved for simulation of non-linear differential equations in time is given by Lemma~\ref{lemma:nogo2}.  Thus the quantum algorithm is near-optimal in general.

\subsection{Continuity with Hamiltonian Simulation}

\label{sec:cont}
We will notice that our result does not appear to demonstrate the expected behavior for a small nonlinearity. Indeed, our algorithm would result in exponential time scaling, even when our nonlinearity is set to $0$. In this section, we will establish that we may indeed find good algorithms when $\|B\|$ is sufficiently small. 

Assume that $\|B\| < \frac{1}{10e^{\|A\|T}T}$. Then, we propose to separate our time steps into $s$ \emph{equal} slices $t_j < \frac{1}{\|A\|}$, and set each $k_j$ to $k$. As per the previous definitions, $m_j = jk$ in this definition. 

We remind ourselves of the following equation from \ref{eq:bound}, \ref{eq:app}, 

\begin{equation}
     \|F_{m,m+k}(t)\| \leq \left(\frac{e(m+k-1)}{k}\right)^k\frac{\|B\|^k}{\|A\|^k}e^{mt\|A\|}(\|A\|\hat{t}_j \times (e-1))^k,
\end{equation}
and we have then:

\begin{equation}
    \frac{m_st\|A\|}{k} = \frac{(sk)(\frac{T}{s})\|A\|}{k} =\|A\|T
\end{equation}

\begin{align}
    \|F_{m_s,m_s+k}(t)\| &\leq \left(\frac{es}{1+1/s}\frac{\|B\|}{\|A\|}e^{\|A\|T}\frac{T}{s} \|A\| (e-1)\right)^k\\
    &\leq \frac{1}{2^k}.
\end{align}
The error introduced at a given time step is bounded by the norm of what we discard. This will be the sum of the tail $\leq \sum_k^{\infty}\frac{1}{2^j} = \frac{1}{2^{k-1}}$ for our choice of $k$. We require $w = O(\|A\|T)$ time slices to fulfill the condition $t_j < \frac{1}{\|A\|}$. Let ${F'}_j$ be the matrix we apply at time step $j$, where $F_j$ is the exact solution. Given an error of at most $\delta$ at each time step, and our steps are multiplicative, the global error introduced by our approximate Carleman linearization will be bounded by

\begin{equation}
    \|({F'}_j\ket{u(0)}-{F}_j)\ket{u(0)}\| \leq (1+\delta)^w.
\end{equation}
If we set $\delta \leq \frac{\epsilon}{2w^2}$, we may bound this by:

\begin{equation}
    \|({F'}_j\ket{u(0)}-{F}_j)\ket{u(0)}\| \leq 2w\delta \leq \epsilon.
\end{equation}
Thus, setting $k = O(\log(\frac{1}{2\epsilon w^2})$, we may obtain an error of at most $\epsilon$. As each time step requires $k$ Carleman blocks, we require a Carleman matrix with $kw = O(\frac{\|A\|T\log(\|A\|\|T\|}{\log(\epsilon)})$ blocks, which will have a norm of at most $O(\frac{\|A\|^2T^2\log(\|A\|\|T\|}{\log(\epsilon)})$ and sparsity $O((d_A+d_B)\frac{\|A\|T\log(\|A\|\|T\|}{\log(\epsilon)})$, as these both grow linearly with the number of Carleman blocks. We may then proceed using standard Hamiltonian simulation techniques to simulate our approximate Hamiltonian with a linear dependence on norm and sparsity of our approximate Hamiltonian, such as, for example, in \cite{7354428}. 

\subsection{Dequantized Algorithm} 
\label{sec:classical}

A question remains about the gap between randomized classical algorithms and quantum algorithms.  We aim to address this now by replicating the process for classical computation and analyzing the scaling difference we can expect. Difficulties arise in the conversion of quantum states into a classical setting; we cannot efficiently represent the Hilbert space in which the states reside. Upon measurement, we obtain an n-dimensional vector; hence our objective will be to find its expected value.

We cannot give our classical algorithm any direct access to the input as this would be too ``powerful"; allowing us to circumvent the lower bounds set by Lemma \ref{lemma:nogo2}, which relies on the difficulty of state discrimination. We will instead allow our classical algorithm to sample from the input. Given an input state:

\begin{equation}
    \ket{\phi} = \sum_j e^{ik_j}c_j\ket{j},
\end{equation}
such that $0 \leq k_j \leq 2\pi$, $c_j \geq 0$ and $\sum_j c_j^2 =1,$ let our classical oracle $O_c$ output the pair $\{j,k_j\}$ with probability $c_j^2$. We note here that this is still not an ideal comparison to the oracle provided in the context of our quantum algorithm; as it solves the phase estimation problem for free. However, it is sufficient for our lower bound to apply; indeed, in the lower bound arguments introduced in Section \ref{sec:nogo}, the states we wish to distinguish are oriented vertically on the Grover sphere initially, and thus the phase being provided does not make the problem easier.  The reader may notice here that equivalently, we may come up with a scheme in which the $c_j$ are directly provided but the $k_j$ are obfuscated, for which the lower bounds would still apply; access to a combination of the two is required for the bound to be broken.

In the context of standard Hamiltonian Simulation, the quantum advantage that we may obtain for certain types of Hermitian operators often relies on a large Hilbert space. The Path Integral Monte-Carlo method may be used to circumvent these constraints in a probabilistic manner, but it generally scales poorly in other variables, such as the sparsity of the operator, the norm of our Hermitian operators, and/or evolution time. 

Determining the exact relation between quantum and classical scaling does not follow clearly for the quadratic differential equation problem. The exponential complexity in time is distinct from how we generally compare quantum and classical simulation; as previously stated, the bound does not exist in the context that classical algorithms are usually defined (access to a deterministic oracle to the input); which would give us the false impression that the classical model has an advantage.

While the partially probabilistic oracle introduced above solves this issue, we note that the Carleman Linearization technique we proposed for the Quantum case creates a linear operator with sparsity and norm scaling exponentially with the time constraint; simple bounds on the path integral Monte Carlo method on such a linear operator would lead to super-exponential scaling in time. Alternative methods may thus be necessary.

First, in order to understand the difference between the performance of quantum and classical algorithms for non-linear differential equations more clearly we need to understand the relationship between the best time scaling achievable for quantum and classical methods.  To this end, we provide a classical method that uses sampling that retains many of the same features as the quantum case including saturating the optimal time scaling of the quantum algorithm, although it may scale poorly in other variables in generality.  We provide the argument below.

\begin{lemma}
\label{euler}
Suppose we have for bounded matrices $A$ and $B$ an equation of the form:

    $$\frac{d\ket{u}}{dt} = A\ket{u} + B\ket{u}^{\otimes2}$$
    Where $A$ and $B$ are $d-sparse$, $\ket{u}$ is a norm 1 vector in $\mathbb{C}^{2^n}$ and the assumption is made that the linear transformation $V(t)u\ket{u}$ on $\ket{u}$ remains unitary at all times $t$. Then there exists a classical algorithm which outputs $v(j)$ such that:
        $$\|v(j)-\bra{u(0)}V(T)\ket{u(0)}\| \leq \epsilon$$
    for a time $T>0$ with respect to the Euclidean norm, using at most

    $$O\left(\frac{d2^nTe^{2(\|A\|+\|B\|)T})}{\epsilon^2}\right)$$   
    queries to classical oracles $O_A, O_B$ such that $O_A(i,j)$ yields the value of the $i^{\rm th}$ non-zero matrix element in row $j$ and similarly $O_B(i,j)$ yields the corresponding matrix element in row $j$ of $B$ and oracles $f_A,f_B$ such that $f_A(x)=y$ if and only if $A_{xy}$ is the $i^{\rm th}$ non-zero matrix element in row $x$ and $f_B$ yields the corresponding locations of the non-zero matrix elements of $B$.
\end{lemma}

\begin{proof}
    Our proof for this claim immediately follows from using the Euler method to approximate the solution to the non-linear differential equation.  Specifically, if  $T/h$ steps are used in the numerical solution to the differential equation then the error in the solution is \cite{hildebrand1987introduction}

\begin{equation}
    |u(t_i) - u(T)| \leq \frac{hMe^{LT}}{2L},
\end{equation}
where $M = \max |\ddot{u}(t)|$ and $L = \max_t |\frac{d}{du} (Au+Bu^{\otimes 2})|$. As $u(t)$ has a constant norm, this gives us bounds on both $M$ and $L$ of $M,L \leq \|A\| + 2\|B\|$. Let $u(t_i)$ be the Euler approximation at the $i$-th step. Thus, setting $h = \frac{4\epsilon/2}{e^{(|A|+2|B|)T}}$, we obtain:

\begin{align} \label{eq:eulerst}
    |u(t_h) - u(T)| &\leq \frac{\frac{4\epsilon^2/2}{e^{(|A|+2|B|)T}}(\|A\| + 2\|B\|)e^{(\|A\| + 2\|B\|)T}}{2(\|A\| + 2\|B\|)}\\
    & \leq \frac{\epsilon}{2}.
\end{align}
Now, we must account for our initial error. Given $\|\Tilde{u}(0)-u(0)\| \leq \delta$, we may represent then $\Tilde{u}(0)=u(0) +r(0)$ where the remainder $r(t)$ is defined to be the difference between $\Tilde{u}(t)$ and $u(t)$. We know that $\frac{dr}{dt}$ is bounded by $Lr(t)$ by our definition of $L$ as the maximum derivative in $u$; thus, we have the following differential equation:

\begin{equation}
    \frac{dr}{dt}\leq Lr(t), \quad r(0) = \delta.
\end{equation}
Integrating the above bound on the derivative we obtain
\begin{equation} \label{eq:initialerr}
    r(T) \leq \delta e^{LT}.
\end{equation}
Noting here that $L\leq \|A\| + 2\|B\|$ and setting here $\delta = \frac{\epsilon}{2e^{(\|A\| + 2\|B\|)T}}$ we obtain that $\|u'(0)-u(0)\| \leq \frac{\epsilon}{2}$. From Equation \ref{eq:eulerst} and Equation  \ref{eq:initialerr} we then have

\begin{equation}
    \|\Tilde{u}(t_h) - u(T)\| \leq \|\Tilde{u}(t_h)- \Tilde{u}(T)\| + \|\Tilde{u}(T) - u(T)\| = \|\Tilde{u}(t_h)- \Tilde{u}(T)\| + \|r(T)\| \leq \epsilon/2 + \epsilon/2 = \epsilon
\end{equation}
Which is the desired precision bound. 

To achieve our initial error of $\delta = \frac{\epsilon}{2e^{(\|A\| + 2\|B\|)T}}$, we will need $O(\frac{e^{2(|A|+2|B|)T}}{\epsilon^2})$ queries. Performing 1 Euler step by directly querying all values will be $O(d2^n)$ queries, and we require $T/h$ Euler steps. This leads to the total scaling of $O(\frac{d2^nTe^{2(|A|+|B|)T}}{\epsilon^2})$ queries for the problem. 
\end{proof}

This close to matches our quantum algorithm, and quantum lower bound in its time scaling. The error scaling is quadratically worse than our quantum algorithm, when an appropriate comparison is made (see Corollary \ref{corollary:errorScaling}). However, the scaling in the dimension is potentially exponentially worse. This is the result of having to compute $f(t,u(t) = Au(t)+Bu(t)^{\otimes2}$ for each of our iterations. 

We may note here, however, that our initial vector only has $O(\frac{e^{2(|A|+2|B|)T}}{\epsilon^2})$ non-zero entries. While this would result in superexponential scaling in generality (as the number of non-zero elements will be multiplied by up to $(d_A+d_B+1)$ at each Euler step), there are cases for which the growth is more reasonable. The simplest case is such that $A$ is diagonal and $B$ is a diagonal block matrix with each block $B_{j,j}$ being a length $n$ vector, and $B_{j,k} =0$ for $k \neq j$. In this case, the number of non-zero entries remains fixed at each Euler step. More broadly, we will be able to define the following class of equations which we can more efficiently solve in a classical setting.

\begin{definition}[Geometrically Local Equation]
\label{def:gemloceq}

Given an equation $\frac{du}{dt} = f(u)$ such that $u \in \mathbb{C}^{2^n}$ and $f:\mathbb{C}^{2^n}  \to \mathbb{C}^{2^n}$ . Let $\ket{x}$,$\ket{y}$ be elementary basis vectors in $\mathbb{C}^{2^n}$.
\begin{enumerate}
    \item We say that $\ket{x}$, $\ket{y}$ \emph{interact} if there exists $u \in \mathbb{C}^{2^n}$, $c \in \mathbb{C}$ such that $\bra{x}f(\ket{u}) \neq \bra{x}f(\ket{u} + \ket{y})$.   
    \item We say that the equation is \emph{geometrically $k$-local in a graph $G$} if we may map all basis vectors in $\mathbb{C}^{2^n}$ onto the vertex set of a graph $G$ in such a way that for two basis vectors $\ket{x}$ and $\ket{y}$ that do not interact, the graph distance $d(x,y)$ is strictly greater than $k$.
\end{enumerate}
\end{definition}

We may note here that this is a different definition to that of a geometrically local Hamiltonian, from the quantum setting. We borrow the terminology as it shares the intuition of geometric locality.

Definition \ref{def:gemloceq} will allow us to bound the number of elements we need to consider in the Euler method. Specifically, given $\tilde{u}(t_{y+1}) = \tilde{u}(t_y) + (T/h)f(\tilde{u}(t_y))$, as per the Euler method. We may see that if $\tilde{u}_j(t_{y+1})$ is non-zero, then there must be some $k$ interacting with $j$ such that $\tilde{u}_k(t_{y})$ is non-zero.

\begin{theorem}
\label{thm:geomloc}
Suppose we have for bounded matrices $A$ and $B$ an equation of the form:

    $$\frac{d\ket{u}}{dt} = A\ket{u} + B\ket{u}^{\otimes2},$$
    where $\ket{u}$ is a norm 1 vector in $\mathbb{C}^{2^n}$ and the assumption is made that $\ket{u(t)}$ remains a unit vector at all times $t$. Assume further that our equation is geometrically $k-local$ in $\mathbb{Z}^D$. Then there exists a classical algorithm which outputs $v(j)$ in $\mathbb{C}^{2^n}$ such that:
        $$\|v(j)-\ket{u(T)}\| \leq \epsilon$$
    for a time $T>0$ with respect to the Euclidean norm, using at most

    $$O\left(\frac{(kT)^{D+1}e^{2(D+1)(|A|+2|B|)T}}{(2\epsilon)^{D+2}}\right)$$   
    queries to classical oracles $O_A, O_B$ such that $O_A(i,j)$ yields the value of the $i^{\rm th}$ non-zero matrix element in row $j$ and similarly $O_B(i,j)$ yields the corresponding matrix element in row $j$ of $B$ and oracles $f_A,f_B$ such that $f_A(x)=y$ if and only if $A_{xy}$ is the $i^{\rm th}$ non-zero matrix element in row $x$ and $f_B$ yields the corresponding locations of the non-zero matrix elements of $B$.
\end{theorem}

\begin{proof}
The method is identical to \ref{euler}, we reevaluate the cost of performing our Euler steps. As each Euler iteration is written as $v(j) + f(v(j), t)$, where $f(u(j)) = \ket{u} + B\ket{u}^{\otimes2}$, all non-zero elements of $v(0)$ may be a distance of at most ${jk}$ from a non-zero element in $v(j)$. Our initial vector has at most $O(\frac{e^{2(|A|+2|B|)T}}{\epsilon^2})$ non-zero values. Given that  $f(u) = A\ket{u} + B\ket{u}^{\otimes2}$ is geometrically local in $\mathbb{Z}^D$, the number of non-zero values of $u(j)$ will thus be bound by the following equation.

\begin{equation}
    O(\frac{e^{2(|A|+2|B|)T}}{\epsilon^2}) \times (jk)^D
\end{equation}
Now, summing over all $m = T/h$ Euler steps,

\begin{equation}
    \sum_{j=1}^m O(\frac{e^{2(|A|+2|B|)T}}{\epsilon^2}) \times (jk)^D
\end{equation}

\begin{equation}
    = O(\frac{e^{2(|A|+2|B|)T}}{\epsilon^2}) \times k^D \times  \sum_{j=1}^{Om}  j^D
\end{equation}

\begin{equation}
    = O(\frac{e^{2(|A|+2|B|)T}}{\epsilon^2}) \times k^D \times  O(m^D)
\end{equation}
Noting here that $m\in O(T\frac{e^{(|A|+2|B|)T}}{4\epsilon/2})$:
\begin{equation}
    = O\left(\frac{(kT)^De^{2(D+1)(|A|+2|B|)T}}{(2\epsilon)^{D+2}}\right)
\end{equation}
We note that we need $O(k)$ calls to our oracle $O_A$ and $O_B$ for each non-zero vector value in $u(j)$. Further, the cost of computing our Euler steps here strictly exceeds the cost of creating an initial vector of sufficient precision. The result follows.  
\end{proof}
A notable set of equations that are geometrically $k-local$ in $\mathbb{Z}^D$ are band matrices (allowing $D$ bands of width up to $k$ in $A$ and width up to $2^n(k-1)$ in $B$).

We will now attempt to write a classical algorithm that will scale polynomially with $n$ in generality, by imitating the methods used for our quantum algorithm. As we will see, this leads to even worse scaling with $T$. 

Part of our arguments will depend on the variation in the probabilistic approximation of our target evolution; for this, we will require a numerical value that we can approximate. We will use here the expected value as defined at the end of section \ref{sec:quantum}; $\bra{u(T)}U\ket{u(T)}$, for some unitary $U$.

Our algorithm will produce a random uniform $P(V = v_j) = P(V = v_k)$ for any choice of $j,k$, where $\sum_jv_j \approx \ket{u(T)}$. We will sample from this at random. A trap to avoid here is that randomly sampling $v_j^\dagger Uv_j$ will not average out at $\bra{u(t)}U\ket{u(t)}$; instead, we will need to randomly sample among $v_j^\dagger Uv_k$. This introduces a quadratic factor that we will accept. Call this random variable $V$, and we will let its variation be $\mathbb{V}(V)$. We won't precisely calculate this variance, but instead use it as one of the variables bounding our running time; with a discussion on how large this value could be after the proof.

\begin{theorem}
\label{thm:path}
Suppose we have for bounded matrices $A$ and $B$ an equation of the form:

    $$\frac{d\ket{u}}{dt} = A\ket{u} + B\ket{u}^{\otimes2}$$
    Where $A$ is $d_A$-sparse, $B$ is $d_B-sparse$, $\ket{u}$ is a norm 1 vector in $\mathbb{C}^{2^n}$ and the assumption is made that $\ket{u(t)}$ remains a unit vector at all times $t$. Suppose that the decomposition of our Carleman block matrix for this problem generates a variance of $$\mathbb{V}(V)$$. Given some Unitary matrix $U$; there exists an algorithm which outputs an approximation $v(T)$ of the expectation value such that 
        $$|v(T)-\bra{u(T)}U\ket{u(T)}| \leq \epsilon$$
    with probability at least $2/3$ for a time $T>0$ using at most

    $$O\left(\frac{\mathbb{V}(V)(d_A^2+d_B^2)e^{O(T(\|A\|+\|B\|)(1+\delta))}}{\epsilon^2}\right)$$
    queries to our classical oracles $O_A$, $O_B$ and $O_u$.
\end{theorem}
\begin{proof}
Through graph coloring, we have access to a $O(d^2\log(N))$ decomposition of a size $N$ $d$-sparse hermitian matrix. This will hold for both $A$ and the extension of $B$ into a pair of Hermitian matrices
\begin{equation}
    B^+ = \begin{bmatrix}
     0 && B  \\
    B^\dagger && 0  \\
    \end{bmatrix}
\end{equation}
\begin{equation}
    B^- = \begin{bmatrix}
     0 && B  \\
    -B^\dagger && 0  \\
    \end{bmatrix}  
\end{equation}
This directly induces a $d^2m\log(N)$-sparse decomposition for $A_m$ and $B^{+/-}_m$, from their respective definitions. Our Carleman Block, with a slight abuse of notation placing each block into its correct position, will be described by

\begin{equation}
    G = \sum_{j=1}^w A_j + \frac{1}{2}B^+_j + \frac{1}{2}B^-_j.
\end{equation}
We note that the set of all even $j$ and the set of all odd $j$ act independently; thus, expanding into our $d^2m\log(N)$-sparse decomposition, this may be expressed as 1 sum of $wd^2\log(N)$ matrices $\sum_{j=1}^{wd^2\log(N)} H_j$. Using a trotter formula, our evolution may then be approximated by:
\begin{equation}
    e^{tG} = \left(\prod_{j=1}^{wd^2\log(N)}e^{\frac{t}{r}H_{j}}\right)^r+ O(t^2).
\end{equation}
This allows us to reduce our error linearly with $r$. We will set $r = O(\frac{T}{\epsilon})$ to ensure our error remains sufficiently small. Note that:

\begin{equation} \label{eq:Y}
    e^{I^{\otimes j} \otimes H \otimes I^{\otimes k}} = I^{\otimes j} \otimes e^H \otimes I^{\otimes j}.
\end{equation}
Thus, we may effectively approximate $e^{A_m}$, $e^{B^+_m}$ and $e^{B^-_m}$.
For any $1$-sparse hermitian $H$,  $e^H$ will be 2-sparse. Ignoring zero-valued terms, our product will take the form
\begin{equation}
  e^{tG}= \prod_{k=1}^{p} (N_1(k) +N_2(k)),
\end{equation}
for some appropriate $p$. $N_1(k)$ and $N_2(k)$ each specify a unique neighbor for us to ``continue" with, given a previous position. We will now take our sum outside of the product to arrive at the following equation,
\begin{equation}
\label{equ:paths}
    e^{tG}=\sum_{l_1=0}^{wd}\sum_{l_2 \in N(l_1)}\sum_{l_3 \in N(l_2)}... \sum_{l_p \in N(l_{p-1})} \prod_{k=1}^{p} N_{l_i}(k),
\end{equation}
where we here slightly abuse notation, $l_s \in N(l_{s-1})$ indicating here the set of $l$ with existing $N_{l_i}(k)$ overlapping with the previous set. We call the elements of our sum ``paths". Intuitively, our path ``travels" through our vector space; with a starting coordinate. The elements of the product describe how the path moves (from coordinate $a$ to coordinate $b$) in a given time splice of our evolution.

Now; computing this exactly will not be possible, due to the impractically large number of possible paths; the objective is instead to approximate the expectation, using Montecarlo methods. 
Let $V$ be the random variable of a randomly chosen path multiplied by the number of different paths, applied to our probability vector. Then,

\begin{align}
    \mathbb{V}\left(\frac{\sum_{k=1}^s V}{s}\right) &= \frac{\mathbb{V}(\sum_{k=1}^s V)}{s^2} = \frac{\sum_{k=1}^s \mathbb{V}(V)}{s^2}\nonumber\\
    &= \frac{\mathbb{V}(V)}{s}.
\end{align}
Chebyshev's inequality then gives us for $\sigma^2 = \mathbb{V}\left(\frac{\sum_{k=1}^s V}{s}\right)$,

\begin{equation}
    Pr(|V-\mathrm{E}[V]|\geq k\sigma) \leq \frac{1}{k^2}.
\end{equation}
Given that $\mathrm{E}\left(\frac{\sum_{k=1}^s V}{s}\right) = \mathrm{E}[X]$, we will require $\sqrt{\frac{\mathbb{V}(v)}{s}} = O(\epsilon \sqrt{\delta})$ if we wish to have a probability of $1-\delta$ of being within $\epsilon$ of our answer.
Thus; the number of samples we require will be $O(\mathbb{V}(V))$.  So our result follows by taking $\delta=1/3$.
Taking the cost of individually simulating each path into consideration, we obtain a total cost of
\begin{equation}
    O\left(\frac{\mathbb{V}(V)n(d_A^2+d_B^2)e^{O(Tc(1+\delta))}}{\epsilon^2}\right).
\end{equation}
This completes our proof of theorem \ref{thm:path} after noting $c\in O(\|A\|+\|B\|)$. 
\end{proof}

Defining exact bounds on the variance of the may be difficult, depending strongly on both the exact problem we are working with and the definition of the matrix on which we are defining our expectation value. Generic bounds, using the maximum norm of any path of 1 will be super-exponential in $T$. Indeed, the number of possible paths we have is super-exponential in $T$ as a result of the exponential sparsity of the largest Carleman block in $T$. However, it should be evident that such bounds vastly overestimate the variance we obtain in practice. In the case of small $B$, as introduced in section \ref{sec:cont}, we may find a variance that is exponential in $Tc$, allowing more reasonable time scaling. Different approaches may however be necessary to obtain optimal time and dimension scaling for the classical case in generality.

\section{Conclusion}

We have studied nonlinear differential equations with unitary solutions. We have found a quantum algorithm solving such equations with $O(e^{cT}\log{\frac{1}{\epsilon}})$ oracle queries, for some constant $c$ dependent on the problem. Further we give classical algorithms that may either match the time scaling, or dimension scaling of our quantum algorithm, though not both in generality. The arguments from Section \ref{sec:nogo} show that any algorithm solving such a problem is necessarily exponential in its time complexity, thus, our algorithms approach the known lower bounds for the problem in time complexity.

Our solutions still offer avenues for significant improvement. Notably, the analysis of the generalized classical algorithm is unlikely to be tight without an in-depth analysis of the variance, and the error scaling can likely be improved. A more suitable classical algorithm may indeed bridge the time-scaling gap with our quantum algorithm. Further, our quantum algorithm could potentially be improved by recent improvements in Hamiltonian simulation, such as the Hamiltonian simulation by qubitization algorithm presented in \cite{Low2019hamiltonian}. These could offer a reduction in the factor of our time exponent if successfully implemented to simulate our linear approximation. The particular formulation of the near-hermitian matrix we obtain for our problem however makes such an implementation non-trivial. 

Establishing a tighter lower bound on the complexity of simulating unitary nonlinear differential equations in terms of both the error and the evolution time would also be an important contribution. This would allow a more direct comparison of the efficiency of our algorithm with complexity theoretic bounds. 

One may want to ask if there exist any families of genuinely nonlinear differential equations under which a quantum algorithm is provably capable of providing exponential speedup over all possible classical algorithms. Given that the linear Hamiltonian simulation problem is a special case of the unitary polynomial equations that we studied, trivially, we may declare the statement to be true; but such an analysis does not demonstrate the capability of quantum computers to handle non-linearity. The work of \cite{Liu_2021} and \cite{https://doi.org/10.48550/arxiv.2202.01054} provide quantum algorithms capable of solving \emph{dissipative} quadratic differential equations in polynomial time under certain restrictions. However, the argument for a quantum advantage of the algorithm from \cite{Liu_2021} over a classical counterpart, once again, uses the special case where the norm of the nonlinearity is set to 0. 

It remains an interesting open problem whether there exists any family of strictly nonlinear differential equations in which a quantum algorithm may provide a provable exponential advantage over any classical algorithm. Ideally, it would be interesting to either:

\begin{enumerate}
    \item Demonstrate that a quantum algorithm may offer no exponential advantage over comparable classical counterparts for any family of strictly quadratic nonlinear differential equations (meaning that $\frac{du}{dt} = F_2u^{\otimes2}$ for a rectangular matrix $F_2$).
    \item Provide a classical lower bound on the running time of a family of strictly quadratic nonlinear differential equations where an exponentially faster quantum algorithm exists.
\end{enumerate}
 Providing clear evidence of such a family would not only improve our understanding of the gap between quantum and classical algorithms for non-linear differential equations but also provide us with a crucial understanding of the extent to which quantum computers will be capable of helping us understand non-linear dynamical systems.

\section*{Acknowledgements}

N.B. is funded by the Department of Computer Science of the University of Toronto, the Ontario Graduate scholarship Program as well as NSERC. N.W. acknowledges funding from Google Inc.  This material is based upon work supported by the U.S. Department of Energy, Office of Science, National Quantum Information Science Research Centers, Co-design Center for Quantum Advantage (C2QA) under contract number DE-SC0012704 (PNNL FWP 76274). 

\bibliographystyle{plainnat}
\bibliography{ref}
\end{document}